\documentclass[floatfix, noshowpacs, preprintnumbers, twocolumn, amsmath, amssymb, aps, superscriptaddress]{revtex4}

\usepackage{setspace}
\usepackage{bm}
\usepackage{amsmath}
\usepackage{amsthm}
\usepackage{amssymb}
\usepackage{graphicx}
\usepackage{bbold}
\usepackage{url}
\usepackage[caption=false]{subfig}
\newcommand{\bra}[1]{\left< #1 \right|} 
\newcommand{\ket}[1]{\left| #1 \right>} 
\usepackage[colorlinks]{hyperref}
\usepackage{hypcap}
\usepackage{mathtools}

\mathtoolsset{centercolon}

\newtheorem{lemma}{Lemma}
\newtheorem{theorem}{Theorem}
\theoremstyle{definition}
\newtheorem{example}{Example}

\newcommand{\eq}[1]{\hyperref[eq:#1]{Eq.~(\ref*{eq:#1})}}
\newcommand{\eqs}[2]{Eqs.\ \hyperref[eq:#1]{(\ref*{eq:#1})} and \hyperref[eq:#2]{(\ref*{eq:#2})}}
\renewcommand{\sec}[1]{\hyperref[sec:#1]{Section~\ref*{sec:#1}}}
\newcommand{\fig}[1]{\hyperref[fig:#1]{Figure~\ref*{fig:#1}}}
\newcommand{\thm}[1]{\hyperref[thm:#1]{Theorem~\ref*{thm:#1}}}
\newcommand{\lem}[1]{\hyperref[lem:#1]{Lemma~\ref*{lem:#1}}}
\newcommand{\ex}[1]{\hyperref[ex:#1]{Example~\ref*{ex:#1}}}

\DeclareMathOperator{\SO}{SO}
\DeclareMathOperator{\poly}{poly}
\DeclareMathOperator{\diag}{diag}

\newcommand{\swap}{\textsc{swap}}
\newcommand{\fswap}{\textrm{f-}\swap}
\newcommand{\iswap}{\textrm{i-}\swap}
\newcommand{\fs}{\textrm{f}\textsc{s}}
\newcommand{\is}{\textrm{i}\textsc{s}}
\newcommand{\cz}{\textsc{cz}}
\newcommand{\A}{\mathcal{A}}
\newcommand{\B}{\mathcal{B}}

\begin{document}
\title{The computational power of matchgates \\ and the XY interaction on arbitrary graphs}

\author{Daniel J.\ Brod}
\email{brod@if.uff.br}
\affiliation{Instituto de F\'isica, Universidade Federal Fluminense, Av.\ Gal.\ Milton Tavares de Souza s/n, 
Gragoat\'a, Niter\'oi, RJ, 24210-340, Brazil}
\affiliation{Institute for Quantum Computing, University of Waterloo, 200 University Ave.\ W., Waterloo, ON, N2L 3G1, Canada}

\author{Andrew M.\ Childs}
\email{amchilds@uwaterloo.ca}
\affiliation{Institute for Quantum Computing, University of Waterloo, 200 University Ave.\ W., Waterloo, ON, N2L 3G1, Canada}
\affiliation{Department of Combinatorics \& Optimization, University of Waterloo, 200 University Ave.\ W., Waterloo, ON, N2L 3G1, Canada}

\date{\today}

\begin{abstract}
Matchgates are a restricted set of two-qubit gates known to be classically simulable when acting on nearest-neighbor qubits on a path, but universal for quantum computation when the qubits are arranged on certain other graphs. Here we characterize the power of matchgates acting on arbitrary graphs. Specifically, we show that they are universal on any connected graph other than a path or a cycle, and that they are classically simulable on a cycle. We also prove the same dichotomy for the XY interaction, a proper subset of matchgates related to some implementations of quantum computing.
\end{abstract}

\maketitle

\section{Introduction} \label{sec:intro}

Studying the computational power of restricted sets of operations can shed light on the nature of quantum speedup. From a theoretical perspective, such studies can determine what resources are necessary and/or sufficient for universal quantum computation. This issue is also relevant in experimental settings, where available operations or resources may be restricted.

In this paper, we focus on the class of operations known as matchgates. Matchgates are a class of $2$-qubit gates defined by Valiant \cite{Valiant02} that are closely related to noninteracting fermions \cite{Terhal02}. To define matchgates, let $G(A,B)$ denote the unitary gate that acts as unitaries $A$ and $B$, respectively, on the even- and odd-parity subspaces of a 2-qubit Hilbert space:
\begin{equation} \label{eq:Matchgate}
G(A,B) = \begin{pmatrix}
A_{11} & 0 & 0 & A_{12} \\
0 & B_{11} & B_{12} & 0 \\
0 & B_{21} & B_{22} & 0 \\
A_{21} & 0 & 0 & A_{22}
\end{pmatrix}.
\end{equation}
The gate $G(A,B)$ is a \emph{matchgate} if $\det A = \det B$.

As originally shown by Valiant \cite{Valiant02}, and soon after by Terhal and DiVincenzo \cite{Terhal02} in the setting of fermionic linear optics, a quantum computation composed only of (i) qubits (arranged on a path) initially prepared in a product state, (ii) a circuit of nearest-neighbor matchgates, and (iii) a final measurement in the computational basis can be efficiently simulated on a classical computer. Curiously, the computational power of matchgates varies greatly with seemingly small changes in spatial restrictions: by relaxing the nearest-neighbor condition and allowing matchgates to also act on next-nearest neighbors, they become universal for quantum computation, as shown by Kempe, Bacon, DiVincenzo, and Whaley \cite{Kempe01b,Kempe02}. Both regimes were revisited and extended by Jozsa and Miyake \cite{Jozsa08b}, who also provided simpler proofs.

More generally, one can consider matchgates restricted to act on pairs of qubits joined by the edges of any graph. In \cite{Brod12} it was shown that matchgates can implement universal quantum computation on many graphs, such as a complete binary tree, a star, or a path with one extra vertex appended to some point, and the authors suggested that the path might be a pathological instance where matchgates are classically simulable. The authors also left as an open question whether there is a regime of intermediate computational power, between that of classical and quantum computers, such as in recent proposals using  commuting operators \cite{Bremner10} or linear optics \cite{Aaronson11}. 

Here we use ideas from \cite{Brod12} to prove that matchgates are universal on any connected graph other than a path or a cycle. We also adapt previous results \cite{Terhal02, Jozsa08b} to show that matchgates are classically simulable on a cycle. Thus we completely characterize the power of matchgates on arbitrary graphs, resolving the two open questions from \cite{Brod12}.

Furthermore, we consider the computational power of the XY (or anisotropic Heisenberg) interaction acting on the edges of a graph. The XY interaction generates matchgates, so it is non-universal on paths \cite{Terhal02} and cycles. In fact, even this restricted class of matchgates is universal when acting on next-nearest neighbors \cite{Kempe01b,Kempe02}. Here we show that the XY interaction is also universal on any connected graph other than a path or a cycle, so it is as powerful as general matchgates. 

This paper is organized as follows. In \sec{univ_Jozsa} we review the proof of universality of matchgates acting on nearest and next-nearest neighbors on a path, focusing on ideas used in our first main result. In \sec{examples} we present two instructive examples from \cite{Brod12} that lead to the proof, in \sec{univ_arbit_graphs}, that matchgates are universal on any connected graph other than a path or a cycle. In \sec{simul_line} we review the classical simulation of matchgates on a path, and in \sec{simul_cyc} we show how this result can be adapted to provide an equivalent result for matchgates acting on a cycle. Finally, in \sec{XY} we specialize the result of \sec{univ_arbit_graphs} and show that the subset of matchgates known as the XY interaction is also universal on any graph other than a path or cycle. Although this latter result implies the first, we present the results separately as the first proof is easier and develops tools that are useful later, while the simulation using the XY interaction is less explicit. 

\smallskip
\noindent\textbf{Notation and terminology.} Throughout this paper we consider matchgates acting on the edges of a graph, unless stated otherwise, and we refer to ``universal graphs'' as those on which such matchgates are universal. We restrict our attention to connected graphs without loss of generality, as qubits in different components of a general graph cannot interact, so the components can be treated separately. By a universal gate set we mean a set that can simulate a universal quantum computer with at most polynomial overhead in number of operations and number of qubits. We extensively use the concept of encoded universality (see, e.g., \cite{Kempe01b,Kempe02,DiVincenzo00}), where one logical qubit is encoded in two or more physical qubits, so we occasionally denote logical basis states or logical operators that act on an encoded space by a subscript $L$ when there is risk of ambiguity. We also interchangeably refer to a set of quantum gates by their unitary matrices or their generating Hamiltonians, as we will not consider the case of discrete sets of unitaries. 

\section{Universality for arbitrary graphs} \label{sec:match_arbit}

\subsection{Matchgates acting on nearest and next-nearest neighbors} \label{sec:univ_Jozsa}

We begin by giving a simple proof, along the lines of \cite{Jozsa08b}, that matchgates are universal on a path when supplemented by the 2-qubit $\swap$ gate. Consider an encoding of a logical qubit into two physical qubits, given by
\begin{align}
\ket{0}_L = \ket{00},  \notag \\
\ket{1}_L = \ket{11}. \label{eq:evenencoding}
\end{align}

Clearly an encoded single-qubit gate $A_L$ can be implemented simply by applying the matchgate $G(A,A)$ to the pair of physical qubits that encode the logical qubit.

The other requirement for a universal set is an entangling 2-qubit gate, such as the controlled-$Z$ ($\cz$) gate. Consider two adjacent logical qubits encoded in physical pairs labeled $\{1,2\}$ and $\{3,4\}$, respectively. Then a $\cz_L$ between the logical qubits can be implemented simply by a $\cz$ between the neighboring qubits 2 and 3. Note that this is not a matchgate---indeed, no nearest-neighbor matchgate can generate entanglement while preserving the encoding of \eq{evenencoding} \cite{Brod11}, as this would contradict that matchgates are classically simulable when acting on a path. Therefore the entangling gate must be implemented with the aid of some non-matchgate operation. One such example is the sequence
\begin{equation} \label{eq:CZL}
\cz = \fswap \cdot \swap.
\end{equation}
Here $\swap = G(I,X)$ is not a matchgate. The closely related gate 
\begin{equation} \label{eq:fS}
  \fswap := G(Z,X) = \begin{pmatrix}
1 & 0 & 0 & 0 \\
0 & 0 & 1 & 0 \\
0 & 1 & 0 & 0 \\
0 & 0 & 0 & -1
\end{pmatrix}
\end{equation}
is a matchgate that swaps the two qubits and induces a minus sign when both are in the $\ket{1}$ state (so we call it the fermionic $\swap$). In \eq{CZL} we can interpret the $\swap$ as undoing an undesired interchange of the qubits induced by the $\fswap$ during the entangling operation.

We thus conclude that matchgates, when supplemented by the $\swap$, form a universal set. Furthermore, the $\swap$ is only applied on disjoint sets of physical qubits (i.e., $\{2i, 2i+1\}$ for $1 \leq i \leq n$, where $n$ is the total number of logical qubits), so no qubit is swapped more than one position away from its original place \cite{Jozsa08b}. Thus the $\swap$ gate in this construction can be replaced by allowing matchgates to also act on second and third neighbors on the path. In fact, matchgates between only nearest and next-nearest neighbors are already universal, as shown in \cite{Kempe02} and \cite{Jozsa08b} using alternative encodings, where each logical qubit is encoded into $3$ and $4$ physical qubits, respectively.

\subsection{Matchgates acting on arbitrary graphs} \label{sec:examples}

Now suppose that the qubits are arranged on the vertices of a more general (connected) graph, and matchgates can act between every neighboring pair of qubits. Henceforth, we restrict ourselves to interactions between nearest neighbors in these general graphs. In this setting, the result mentioned at the end of the previous section straightforwardly translates to the universality of the ``triangular ladder'' graph of \fig{triangladder} \cite{Kempe02}. 

\begin{figure}
\capstart
\centering
\includegraphics[width=0.3\textwidth]{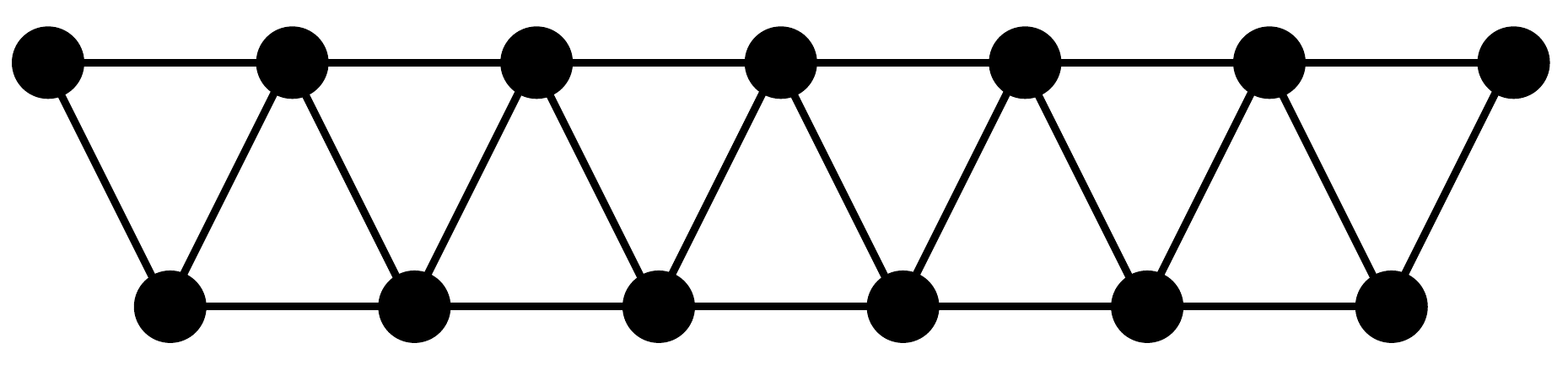}  
\caption{In a triangular ladder graph, vertices have a one-to-one correspondence to vertices of a path such that nearest neighbors on the triangular ladder correspond to nearest and next-nearest neighbors on the path.}	
\label{fig:triangladder}
\end{figure}

In a previous paper \cite{Brod12}, it was proven that matchgates are also universal on many other graphs. Here we extend this result to show that matchgates are universal on any graph that is not a path or a cycle.

Before giving the proof for the most general case, it is instructive to work through two cases that exemplify the main ideas. Both examples are taken from \cite{Brod12} with small adaptations.

\begin{example} \label{ex:path}
Suppose the qubits are arranged according to a graph of the form shown in \fig{appendedline}, which is obtained by joining a new vertex to some degree-2 vertex of a path. To prove that such a graph is universal, we use two tricks from \cite{Brod12}. 

\begin{figure}
\capstart
\centering
\includegraphics[width=0.3\textwidth]{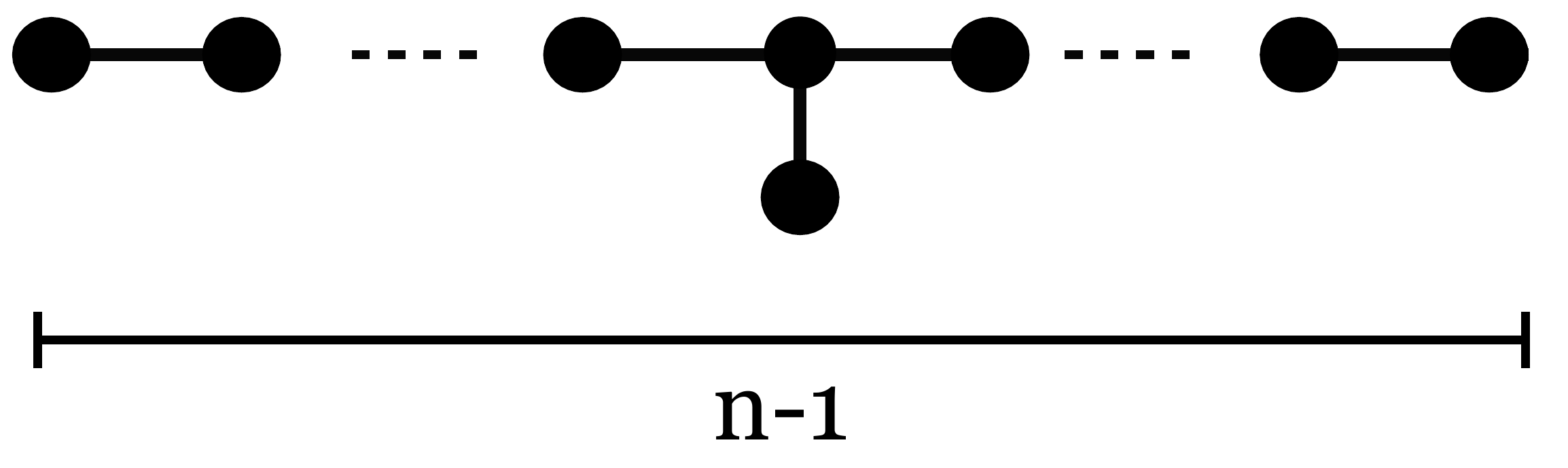}  
\caption{An $n$-vertex graph obtained from an $(n-1)$-vertex path by joining a new vertex to some degree-$2$ vertex of the original path.}	
\label{fig:appendedline}
\end{figure}

First, suppose we have a logical qubit in an arbitrary state $\ket{\Psi}_L=\alpha \ket{00} + \beta \ket{11}$ and a third physical qubit in any state $\ket{\phi}$. We then have the identity
\begin{equation} \label{eq:logicswap}
\fs_{12} \fs_{23} \ket{\Psi}_L \ket{\phi} = \ket{\phi} \ket{\Psi}_L,
\end{equation}
where $\fs$ is shorthand for the $\fswap$ gate, and subscripts denote the pair being acted on. The above identity follows from the trivial observation that the logical qubit is always a superposition of $\ket{00}$ and $\ket{11}$, so the $\fswap$ gate either does not induce a minus sign, or does so twice.  Thus, the $\fswap$ can replace the $\swap$ provided it exchanges a complete logical qubit. Note that, by linearity, this holds even if the logical state of qubits $1$ and $2$ is entangled with other logical qubits, as long as it is a physical state of even parity. 

The second trick is the identity
\begin{equation} \label{eq:0swap}
\fs \ket{0} \ket{\psi} = \ket{\psi} \ket{0}
\end{equation}
where $\ket{\psi}$ is the state of any physical qubit.
This follows simply because
when either of the qubits is in the $\ket{0}$ state, the $\fswap$ does not induce a minus sign, behaving exactly as the $\swap$. We will use this fact to initialize some ancilla qubits in the $\ket{0}$ state and move them around as necessary. 

We now prove universality for \ex{path}. First note that the graph of \fig{branching} is guaranteed to appear as a subgraph of the one in \fig{appendedline} if the number of vertices is greater than $6$. We refer to the degree-$3$ vertex in that graph---and more generally, to any vertex of degree greater than $2$ in a tree---as a branching point. We initialize two ancilla qubits near the branching point (specifically, at vertices $\alpha$ and $\beta$ in \fig{branching}) as $\ket{0}$ and encode the logical qubits using pairs of adjacent qubits as in \eq{evenencoding}. Depending on the number of vertices and the location of the branching point, some physical qubits might be unpaired, in which case one or two qubits at the endpoints may not be used.

\begin{figure}[t]
\capstart
\centering
\includegraphics[width=0.3\textwidth]{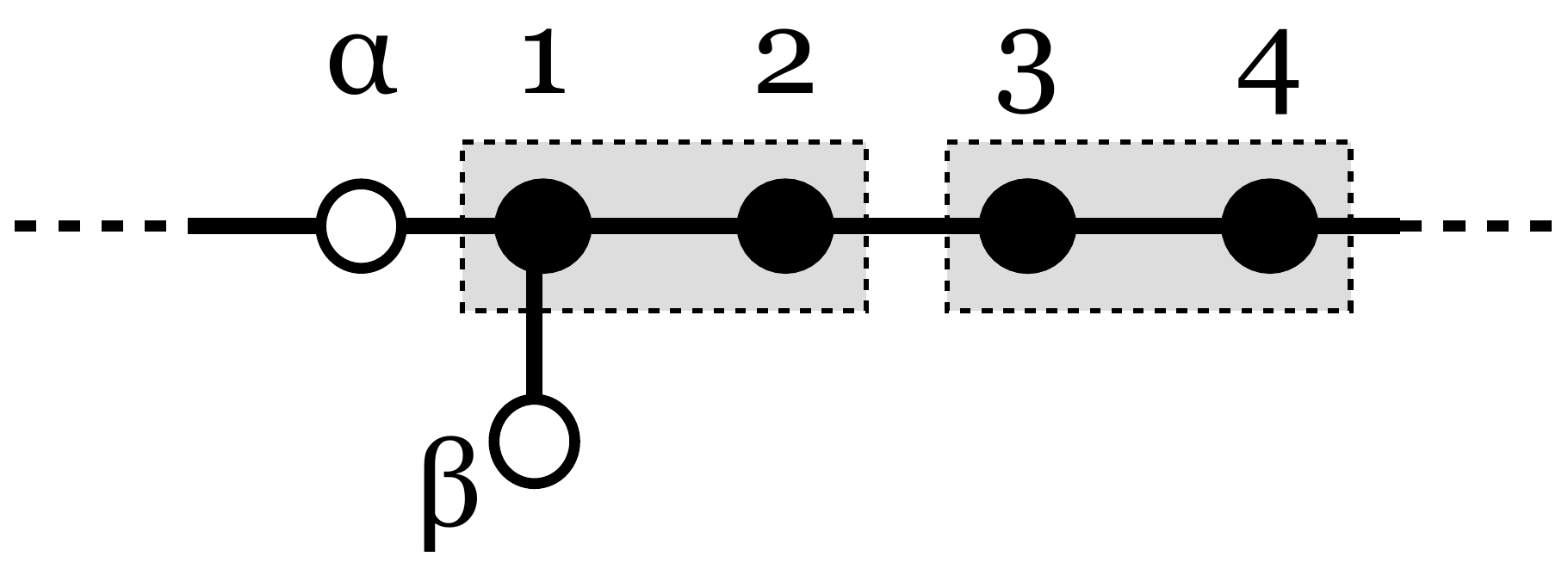}
\caption{Close-up view of the degree-$3$ vertex of the graph in \fig{appendedline}. Vertices labeled $\alpha$ and $\beta$ correspond to ancillas initialized in the $\ket{0}$ state. Vertex pairs $\{1,2\}$ and $\{3,4\}$ correspond to the two logical qubits on which we want to implement a logical $\cz$ gate. The $\alpha$ and $\beta$ ancillas are used to change the order of the state of the other qubits, as per \eq{switch}.}
\label{fig:branching}
\end{figure} 

As discussed above, a logical single-qubit gate $A$ can be implemented simply by a $G(A,A)$ matchgate between adjacent qubits. Now say we want to implement a logical $\cz$ gate between two (not necessarily adjacent) logical qubits. We first use the identity of \eq{logicswap} to place the two desired pairs near the branching point, as in \fig{branching}. In the previous section, we mentioned that the logical $\cz$ can be implemented by a physical $\cz$ between two of the four qubits (e.g., $1$ and $3$, as labeled in \fig{branching}), which in turn is equal to $\swap$ followed by $\fswap$. We can implement this sequence by swapping qubit $2$ with both qubits of the pair ($3$,$4$), which is possible by \eq{logicswap}, and then using the fact that $\alpha$ and $\beta$ are ancillas in the $\ket{0}$ state to switch the order of the qubits placed in vertices $1$ and $2$. This effectively implements the $\swap$ of \eq{CZL}. If we follow this with an $\fswap$ again between qubits $1$ and $2$, the final result is the desired $\cz$ gate. We can then use \eq{logicswap} to return all qubits to their original places. The explicit sequence is
\begin{equation} \label{eq:switch}
\fs_{23} \, \fs_{34} \, \fs_{12} \, \fs_{\beta 1} \, \fs_{12} \, \fs_{\alpha 1} \, \fs_{\beta 1} \, \fs_{12} \, \fs_{\alpha 1} \, \fs_{34} \, \fs_{23}.
\end{equation} 

This sequence uses only matchgates to implement a $\cz$ between the logical qubits which, together with the single-qubit gates mentioned previously, gives a universal set. Since any logical qubit can be moved to any desired location using $O(n)$ $\fswap$ gates, the overhead in the number of such gates grows polynomially with the number of 2-qubit gates in the original circuit.
\end{example}

\begin{example} \label{ex:leaves}
Now suppose the qubits are arranged on a complete binary tree of $m$ levels, as in \fig{binarytree}. This graph has  $n = 2^{m+1}-1$ vertices. Since the longest path contains only $2m-1=O(\log n)$ vertices, the strategy of \ex{path} cannot be trivially adapted to this case: the number of available logical qubits would not be sufficient.

\begin{figure}
\capstart
\centering
\includegraphics[width=0.35\textwidth]{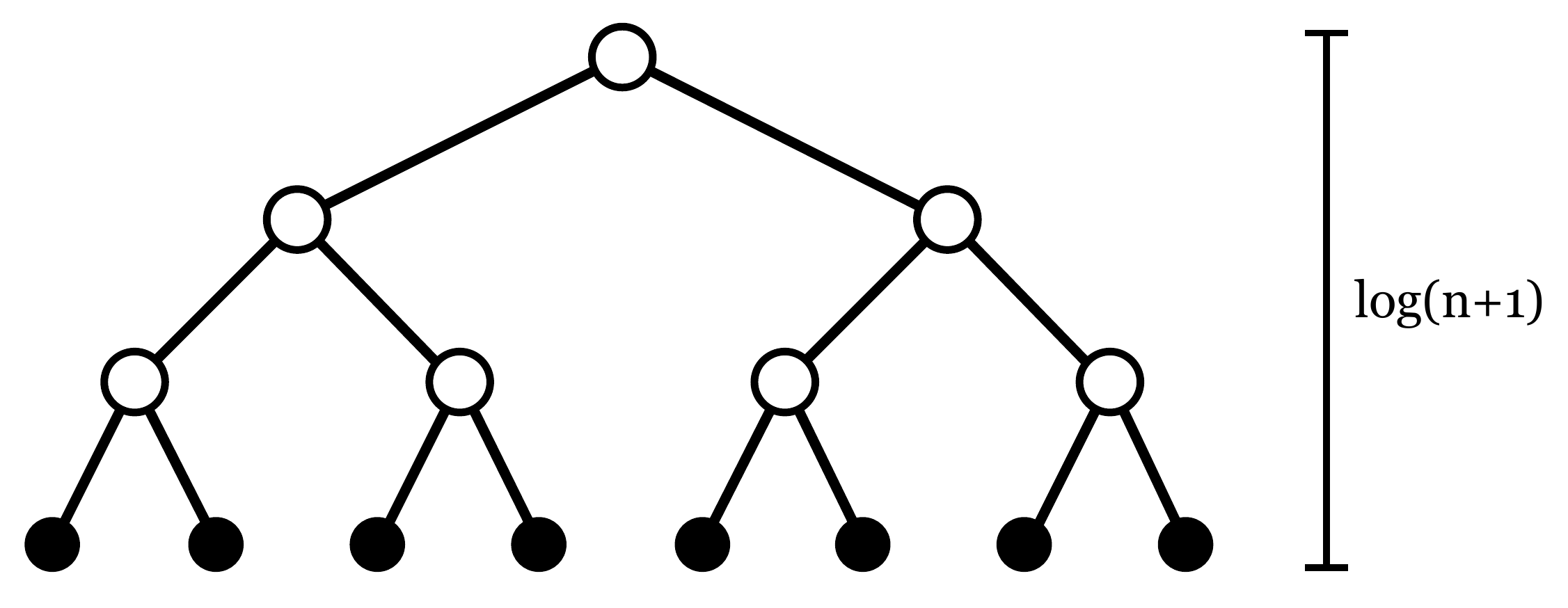}  
\caption{An $n$-vertex complete binary tree. White vertices represent $\ket{0}$ ancillas and black vertices are used in pairs to store computational qubits. This arrangement enables universal computing with matchgates.}
\label{fig:binarytree}
\end{figure}

Instead, we store logical qubits using the $2^m=(n+1)/2$ leaves as shown in \fig{binarytree}. By using the leaves as the computational qubits and filling the paths that connect them with $\ket{0}$ ancillas, we can use the identity of \eq{0swap} to move the state of any qubit to a vertex adjacent to any other desired qubit in less than $2 \log(n/2)$ steps, apply the desired matchgate between them, and return them to their initial positions. This means we can use the $\fswap$ to implement an effective interaction between any pair among the $(n+1)/2$ computational qubits, which clearly is sufficient for universal computation, as per the construction of \sec{univ_Jozsa}. The overhead of this approach is modest: it requires twice the number of qubits and uses $2 \log(n)$ $\fswap$ operations per 2-qubit gate in the original circuit. Note that this approach works for any pairing of physical into logical qubits.
\end{example}

\subsection{Main result} \label{sec:univ_arbit_graphs}

The two examples of the previous section provide the main ideas for a complete characterization of the power of matchgates on arbitrary graphs. To obtain this result, we first prove the following lemma:

\begin{lemma} \label{lem:graph}
Let $T$ be an $n$-vertex tree with $l$ leaves and a longest path of length $p$. Then either (i) $l > \sqrt{n}$ or (ii) $p > \sqrt{n}$.
\end{lemma}

\begin{proof}
Choose any leaf $v$ of $T$. Delete every vertex on the path from $v$ to the nearest branching point, not including the branching point (see \fig{arbitree}). Since, by hypothesis, this path has length smaller than $p$, the result is a subtree of $T$ where one leaf and at most $p-1$ vertices are removed. Repeat this procedure until only a path remains (i.e., $l-2$ times). Finally, delete the remaining path, removing the last two leaves and at most $p$ vertices. This process deletes every vertex in $T$. Therefore $n \leq (l-2)(p-1)+p < lp$, so $\max\{l, p\} > \sqrt{n}$ as claimed.
\end{proof}

\begin{figure}
\capstart
\centering
\includegraphics[width=0.3\textwidth]{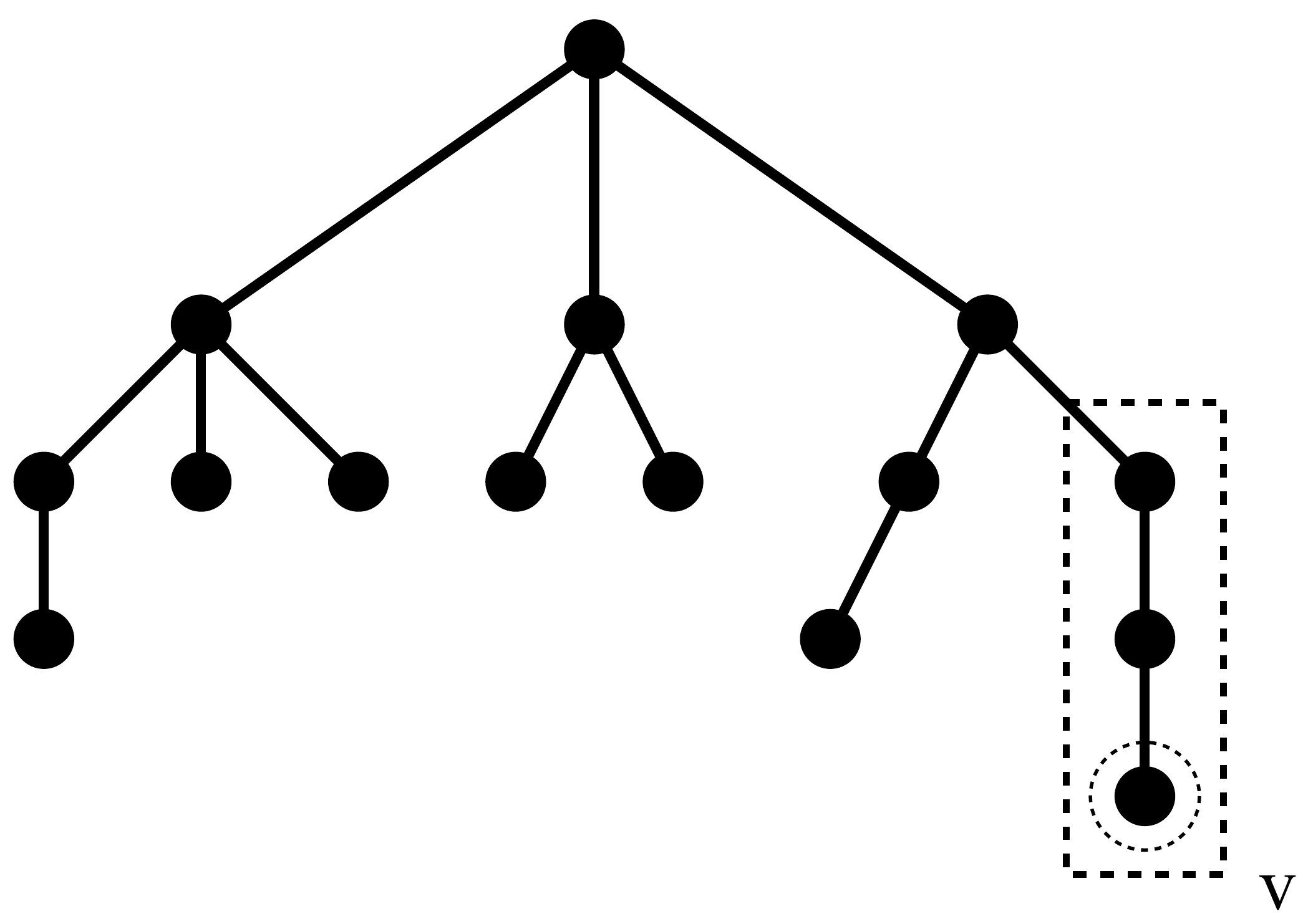}  
\caption{A tree. The dashed rectangle indicates the vertices in the path from $v$ to the nearest branching point, which are deleted in the proof of \lem{graph}. Upon deletion of these vertices, the remaining tree has one fewer leaf, and at most $p-1$ vertices have been removed.}	
\label{fig:arbitree}
\end{figure}

The main result follows straightforwardly from \lem{graph} and the examples of the previous section:

\begin{theorem} \label{thm:maintheo}
Let $G$ be any $n$-vertex connected graph, other than a path or a cycle, where every vertex represents a qubit and we can implement arbitrary matchgates between neighbors in $G$. Then it is possible to efficiently simulate (i.e., with polynomial overhead in the number of operations) any quantum circuit on $\Omega(\sqrt{n})$ qubits.
\end{theorem}

\begin{proof}
Since $G$ is not a path or a cycle, it has some spanning tree $T$ that is not a path. This holds because $G$ necessarily contains a vertex of degree more than 2 and one can construct a spanning tree that includes all edges adjacent to this vertex. It suffices to show that universal computation can be implemented in $T$, since all edges of $T$ are edges of $G$. By \lem{graph}, either (i) the longest path of $T$ or (ii) the set of all its leaves must have more than $\sqrt{n}$ vertices. 

First, suppose (i) holds. Assign each qubit of a longest path of $T$ as a computational qubit, with the exception of one qubit at a branching point. We also use one qubit adjacent to the branching point and not in the path as an ancilla. All other qubits are ignored. We implement the circuit as shown in \ex{path} of the previous section. Since the longest path has more than $\sqrt{n}$ vertices by hypothesis, this allows the simulation of an arbitrary quantum circuit on $\lfloor (\sqrt{n}-1)/2 \rfloor$ qubits. This simulation uses $O(n)$ $\fswap$ operations for each two-qubit gate.

Otherwise (ii) holds, so $T$ has more than $\sqrt{n}$ leaves. Proceed by assigning every qubit at a leaf as a computational qubit and initializing every other qubit as a $\ket{0}$ ancilla. The intermediate vertices on the (unique) path between any two leaves represent qubits in the $\ket{0}$ state. As in \ex{leaves}, we can use the identity of \eq{0swap} to move the state of any qubit to a vertex adjacent to any other, implement a matchgate, and move it back. Thus we can effectively implement any matchgate between any pair of logical qubits. Since the longest path has length less than $\sqrt{n}$, this simulation uses $O(\sqrt{n})$ $\fswap$ operations for each gate in the original circuit.
\end{proof}

\section{Classical simulation of matchgates on the path and cycle} \label{sec:simul}

In the previous section, we proved the universality of matchgates on any connected graph that is not a path or a cycle. We now show that this is also a necessary condition (assuming that quantum computers cannot be efficiently simulated classically).

As mentioned in \sec{intro}, it is well-known that matchgates on a path can be simulated classically for any product state input and computational basis measurement \cite{Valiant02, Terhal02, Jozsa08b}. We briefly review the proof of this fact as shown in \cite{Jozsa08b} and then generalize the proof to the case of a cycle.

\subsection{The Jordan-Wigner transformation and classical simulation of nearest-neighbor matchgates} \label{sec:simul_line}

We begin by defining the Jordan-Wigner operators \cite{Jordan28} acting on $n$ qubits:
\begin{align} \label{eq:JW}
c_{2j-1} &:= \left( \prod_{i=1}^{j-1} Z_i \right) X_j \notag \\
c_{2j} &:= \left( \prod_{i=1}^{j-1} Z_i \right) Y_j
\end{align}
for $j \in \{1,\ldots,n\}$, where $X_i,Y_i,Z_i$ denote the Pauli $X$, $Y$, and $Z$ operators, respectively, acting on qubit $i$. Using this transformation, we can write
\begin{align}
c_{2k-1} c_{2k} & = i Z_k \label{eq:JWtransf1}
\end{align}
for $k \in \{1,\ldots,n\}$ and
\begin{align}
c_{2k} c_{2k+1} & = i X_k X_{k+1} \notag \\
c_{2k-1} c_{2k+2} & = -i Y_k Y_{k+1} \notag \\
c_{2k-1} c_{2k+1} & = -i Y_k X_{k+1} \notag \\
c_{2k} c_{2k+2} & = i X_k Y_{k+1} \label{eq:JWtransf2}
\end{align}
for $k \in \{1,\ldots,n-1\}$. These two-qubit Hamiltonians are precisely the generators of the group of nearest-neighbor matchgates \cite{Terhal02}.

Suppose that the circuit being simulated has an initial product state input $\ket{\psi}=\ket{\psi_1} \ket{\psi_2}\ldots\ket{\psi_n}$, a sequence of nearest-neighbor matchgates, and a final measurement in the computational basis. To simulate the final measurement of qubit $k$, it suffices to calculate of the expectation value $\langle Z_k \rangle$ = $-i \langle c_{2k-1} c_{2k} \rangle = -i \bra{\psi} U^{\dagger} c_{2k-1} c_{2k} U \ket{\psi}$, where $U$ is the unitary corresponding to the action of the matchgate circuit. To show that this can be calculated efficiently, we invoke the following (cf. \cite{Knill01,Terhal02,Jozsa08b}, as stated in \cite{Jozsa08b}):

\begin{theorem} \label{thm:quadratic}
Let H be any Hamiltonian quadratic in the operators $c_i$ and let $U = e^{iH}$ be the corresponding unitary. Then, for all $\mu \in \{1,\ldots,2n\}$,
\begin{equation}
U^{\dagger} c_\mu U = \sum_{\nu=1}^{2n} R_{\mu,\nu} c_\nu,
\end{equation}
where $R\in\SO(2n)$, and we obtain all of $\SO(2n)$ this way.
\end{theorem}

The straightforward proof of this theorem appears in \cite{Jozsa08b}. Observe that, according to \eqs{JWtransf1}{JWtransf2}, the Hamiltonians that generate nearest-neighbor matchgates are quadratic in the operators $c_i$, so
\begin{align} \label{eq:expectedZ}
\langle Z_k \rangle & = -i \bra{\psi} U^{\dagger} c_{2k-1} c_{2k} U \ket{\psi} \\ \notag
& = -i \sum_{a,b=1}^{n} R_{2k-1, a} R_{2k, b} \bra{\psi} c_a c_b \ket{\psi}.
\end{align}

If $t$ is the number of matchgates in the circuit, $R \in \SO(2n)$ can be calculated in $\poly(n,t)$ time as the product of the rotations corresponding to each matchgate. Also notice that the sum in \eq{expectedZ} has only $O(n^2)$ terms. Finally, note that $\ket{\psi}$ is a product state, and any monomial $c_a c_b$ is a tensor product of Pauli matrices, as is clear from \eq{JW}. Thus, each term in the sum factors as a product of the form $\prod_{i=1}^{n} \bra{\psi_i} \sigma_i \ket{\psi_i}$, which involves $n$ efficiently computable terms. Since $\langle Z_k \rangle$  is a sum of a polynomial number of efficiently computable terms, it can be computed efficiently, which completes the proof of classical simulability of matchgates on a path.

\subsection{Classical simulation of matchgates on a cycle} \label{sec:simul_cyc}

The result of the previous section does not immediately apply to the case of a cycle, which corresponds to applying periodic boundary conditions to a path, because a matchgate between the first and last qubits does not translate into a Hamiltonian that is quadratic in the $c_i$s, and vice versa. For example,
\begin{equation} \label{eq:notnnmatch}
c_1 c_{2n} = i X_1 X_n \prod_{i=1}^{n} Z_i,
\end{equation}
which is clearly not a matchgate, as it is a unitary operation acting on every qubit in the circuit. 

Note that \thm{quadratic} still applies to the Hamiltonian in \eq{notnnmatch} even though it does not correspond to a matchgate. However, we do not have a straightforward way of writing the operators we need, such as $X_1 X_n$, in terms of these quadratic operators. 

To show that matchgates are simulable in this case nonetheless, first consider the case where the input state $\ket{\psi}$ is a computational basis state. Suppose that $\ket{\psi}$ has even parity (e.g., $\ket{000\ldots0}$). Matchgates preserve parity, so the state at any point in the computation has a well-defined (even) parity. Now notice that $\prod_{i=1}^{n} Z_i$ is the operator that measures overall parity, so it acts as the identity on the even-parity subspace. This means that for any even-parity input we have the correspondence
\begin{equation} 
X_1 X_n = X_1 X_n \prod_{i=1}^{n} Z_i = - i c_1 c_{2n} \quad \text{(even parity)},
\end{equation}
where the second equality is just \eq{notnnmatch}. The equivalent equations for $Y_1 Y_n$, $X_1 Y_n$, and $Y_1 X_n$ are straightforward. Since we have recovered a correspondence between matchgates on qubits $1$ and $n$ and quadratic Hamiltonians, the simulation can be carried out exactly as in \sec{simul_line}. The case of an odd-parity input state (e.g., $\ket{100\ldots0}$) is analogous, except that the operator $\prod_{i=1}^{n} Z_i$ now acts as minus the identity, and we write
\begin{equation} 
X_1 X_n = - X_1 X_n \prod_{i=1}^{n} Z_i =  i c_1 c_{2n} \quad \text{(odd parity)}
\end{equation}
and its equivalents for $Y_1 Y_n$, $X_1 Y_n$, and $Y_1 X_n$. 

Now consider a general product input state $\ket{\psi}$. Let $\ket{\psi_{\pm}}$ denote the projections of $\ket{\psi}$ onto the even- and odd-parity subspaces, respectively. The expectation value $\langle Z_K \rangle$, analogous to \eq{expectedZ}, is
\begin{align} \label{eq:simulationcycle}
\langle Z_k \rangle = & -i \bra{\psi} U^{\dagger} c_{2k-1} c_{2k} U \ket{\psi} \notag \\ 
= & -i \sum_{a,b=1}^{n} ( R_{2k-1, a} R_{2k, b} \bra{\psi_{+}} c_a c_b \ket{\psi_{+}} \notag \\ 
& \qquad\quad + R'_{2k-1, a} R'_{2k, b} \bra{\psi_{-}} c_a c_b \ket{\psi_{-}} ).
\end{align}
Here $R$ and $R'$ correspond to two sets of rotations, where $R'$ includes an extra minus sign for every matchgate applied between qubits $1$ and $n$. The expression above does not contain cross terms such as $\bra{\psi_{-}} c_{a} c_{b} \ket{\psi_{+}}$ because $c_a c_b$ preserves parity.

The sum in \eq{simulationcycle} contains a polynomial number of terms, just as in \eq{expectedZ}, but now each term may not be easy to compute, since $\ket{\psi_{\pm}}$ are not product states in general. However, we have
\begin{align*}
\bra{\psi} c_{a} c_{b} \ket{\psi} & = \bra{\psi_{+}} c_{a} c_{b} \ket{\psi_{+}} + \bra{\psi_{-}} c_{a} c_{b} \ket{\psi_{-}}, \\
\bra{\psi} c_{a} c_{b} \prod_{i=1}^{n} Z_i \ket{\psi} & = \bra{\psi_{+}} c_{a} c_{b} \ket{\psi_{+}} - \bra{\psi_{-}} c_{a} c_{b} \ket{\psi_{-}}.
\end{align*}
We can invert these equations to obtain
\begin{align}
\bra{\psi_{+}} c_{a} c_{b} \ket{\psi_{+}} & = \frac{1}{2} \left[ \bra{\psi} c_{a} c_{b} \ket{\psi}+\bra{\psi} c_{a} c_{b} \prod_{i=1}^{n} Z_i \ket{\psi} \right ], \notag \\
\bra{\psi_{-}} c_{a} c_{b} \ket{\psi_{-}} & = \frac{1}{2} \left[ \bra{\psi} c_{a} c_{b} \ket{\psi}-\bra{\psi} c_{a} c_{b} \prod_{i=1}^{n} Z_i \ket{\psi} \right ]. \label{eq:expectedparity}
\end{align}

The left-hand sides are precisely the two terms of $\langle Z_k \rangle$ that we need, while the right-hand sides are combinations of terms that can be efficiently computed, as both are expected values of products of Pauli operators on product states. Explicitly, if $\ket{\psi}=\ket{\psi_1}\ket{\psi_2}\ldots\ket{\psi_n}$ and $c_a c_b = \sigma_1 \sigma_2 \ldots \sigma_n$, we have
\begin{align}
\bra{\psi} c_{a} c_{b} \ket{\psi} &= \prod_{i=1}^{n} \bra{\psi_i} \sigma_i \ket{\psi_i}, \\
\bra{\psi} c_{a} c_{b} \prod_{i=1}^{n} Z_i \ket{\psi} &= \prod_{i=1}^{n} \bra{\psi_i} \sigma_i Z_i \ket{\psi_i}.
\end{align}

Plugging \eq{expectedparity} into \eq{simulationcycle}, we recover an expression that can be efficiently computed in the same manner as \eq{expectedZ}, with only four times as many terms. This gives an efficient classical simulation for matchgates acting on a cycle, as claimed.

Note that the simulation scheme of \sec{simul_line} was recently exploited \cite{Jozsa10} to show that circuits of nearest-neighbor matchgates on $n$ qubits are equivalent to general quantum circuits on $O(\log n)$ qubits, and subsequently \cite{Kraus11, Boyajian13} to show a protocol for ``compressed'' simulations (i.e., with quantum circuits on $O(\log n)$ qubits) of the Ising and XY models of spin systems with open boundary conditions. We leave it as an open question whether the observations made in this section lead to analogous results for systems with periodic boundary conditions.

\section{Universality of the XY interaction on arbitrary graphs} \label{sec:XY}

In \sec{match_arbit} and \sec{simul}, we investigated the computational power of the set of all matchgates on arbitrary graphs. We now consider the computational power of a restricted set of matchgates corresponding to the XY (or anisotropic Heisenberg) interaction on arbitrary graphs. This interaction corresponds to a subset of matchgates generated by the Hamiltonian $H := X \otimes X + Y \otimes Y$ (recall from \sec{simul_line} that matchgates are generated by the two-qubit Hamiltonians $X \otimes X$, $X \otimes Y$, $Y \otimes X$, $Y \otimes Y$ together with the single-qubit Hamiltonian $Z$). It is easy to see that these interactions form a proper subset of matchgates as, e.g., they act non-trivially only on the odd-parity subspace of the $2$-qubit Hilbert space.

The Hamiltonian $H$ is an idealized model of the interactions present in several proposed physical implementations of quantum computing, such as quantum dots \cite{Imamoglu99, Quiroga99}, atoms in cavities \cite{Zheng00}, and quantum Hall systems \cite{Mozyrsky01}. A comparison of these proposals can be found in \cite{Lidar01}.

Despite being a proper subset of matchgates, the XY interaction is also known \cite{Kempe02} to be universal for quantum computation when acting on the graph of \fig{triangladder} (i.e., nearest and next-nearest neighbor interactions on a path). It also follows trivially from \sec{simul} that the XY interaction is classically simulable on paths and cycles. This prompts the question of whether our results from \sec{match_arbit} can be adapted for the XY interaction on arbitrary graphs.

In fact, we now show that the XY interaction alone is universal for quantum computation on any connected graph that is not a path or a cycle. Since these operations are a subset of matchgates, this result subsumes the one of \sec{match_arbit}. However, the argument we give for the XY interaction is less explicit, and the simulation is less efficient in general.

First observe that the XY interaction acts trivially on the even-parity subspace, so the encoding of \eq{evenencoding} cannot be used. A suitable alternative (as used in \cite{Kempe02}) is
\begin{align}
\ket{0}_L & = \ket{01}, \notag \\
\ket{1}_L & = \ket{10}, \label{eq:oddencoding}
\end{align}
which is simply the corresponding encoding on the odd-parity subspace. 

We also need to adapt some of the identities used in \sec{match_arbit}. The fermionic $\swap$ gate is not available, so instead we use the following similar gate, which we call the $\iswap$ (and denote by the shorthand $\is$):
\begin{equation} \label{eq:i-SWAP}
\is := \exp( i \tfrac{\pi}{4} H ) = G(I, iX) = \begin{pmatrix}
1 & 0 & 0 & 0 \\
0 & 0 & i & 0 \\
0 & i & 0 & 0 \\
0 & 0 & 0 & 1
\end{pmatrix}.
\end{equation}

For an arbitrary logical state $\ket{\Psi}_L = \alpha \ket{10} + \beta \ket{01}$ in the encoding of \eq{oddencoding}, and for any physical qubit in an arbitrary state $\ket{\phi}$, we have the following identity (already used implicitly in \cite{Kempe02}):
\begin{equation} \label{eq:logicswap2}
\is_{12} \, \is_{23} \ket{\Psi}_L \ket{\phi} = i \ket{\phi} \ket{\Psi}_L.
\end{equation}
Thus these states can be swapped up to an irrelevant global phase. 

Another useful identity, akin to \eq{0swap}, is given by 
\begin{equation} \label{eq:0iswap}
\is_{12}\ket{0} \ket{\psi} = \left ( P \ket{\psi} \right ) \ket{0},
\end{equation}
where $\ket{\psi}$ is any state and $P := \diag(1,i)$. This identity has a familiar operational interpretation: once more any state can be ``swapped through'' a $\ket{0}$ ancilla, but now with the caveat that the state suffers an unwanted $P$ gate. We must take this into account when using \eq{0iswap} in a simulation, but one can already see that if we only need to swap states through an even number of ancillas at a time, we can cancel out the $P$ gates by alternating $\iswap$ and $\iswap$$^\dagger$ swapping operations. In fact, a trivial adaptation of \thm{maintheo} gives a proof of universality for those graphs that have an odd cycle (i.e., non-bipartite graphs), since then there is always an even-length path between any two vertices. We state this without proof, as the details are not instructive and the result is implied by the general case. Note however that for non-bipartite graphs, one can obtain a universal set of unitary matrices, whereas for general graphs we will only obtain a universal set of orthogonal matrices.

\begin{figure}[t]
\capstart
\centering
\subfloat[]{\includegraphics[width=0.13\textwidth]{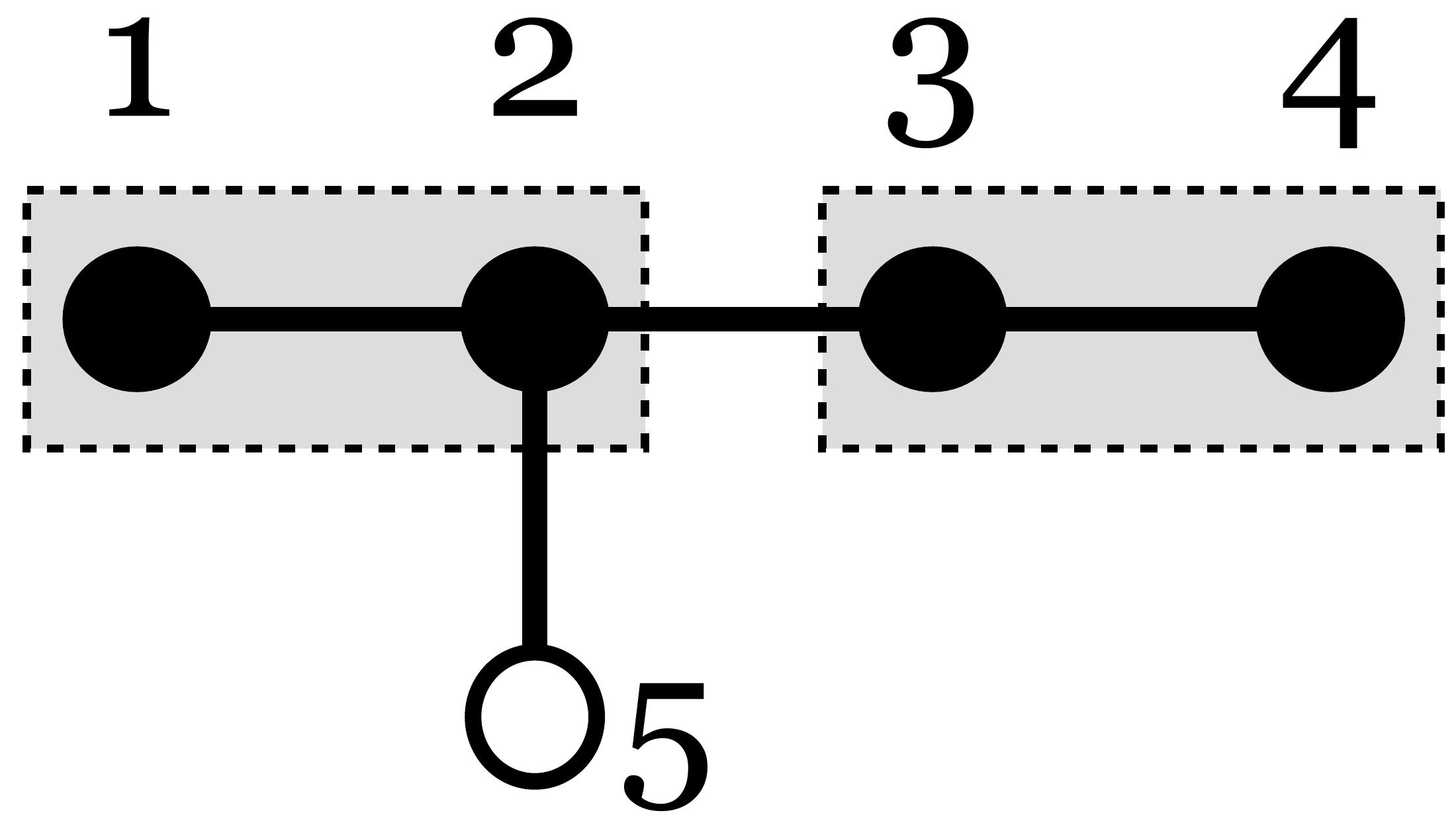}} \qquad
\subfloat[]{\includegraphics[width=0.13\textwidth]{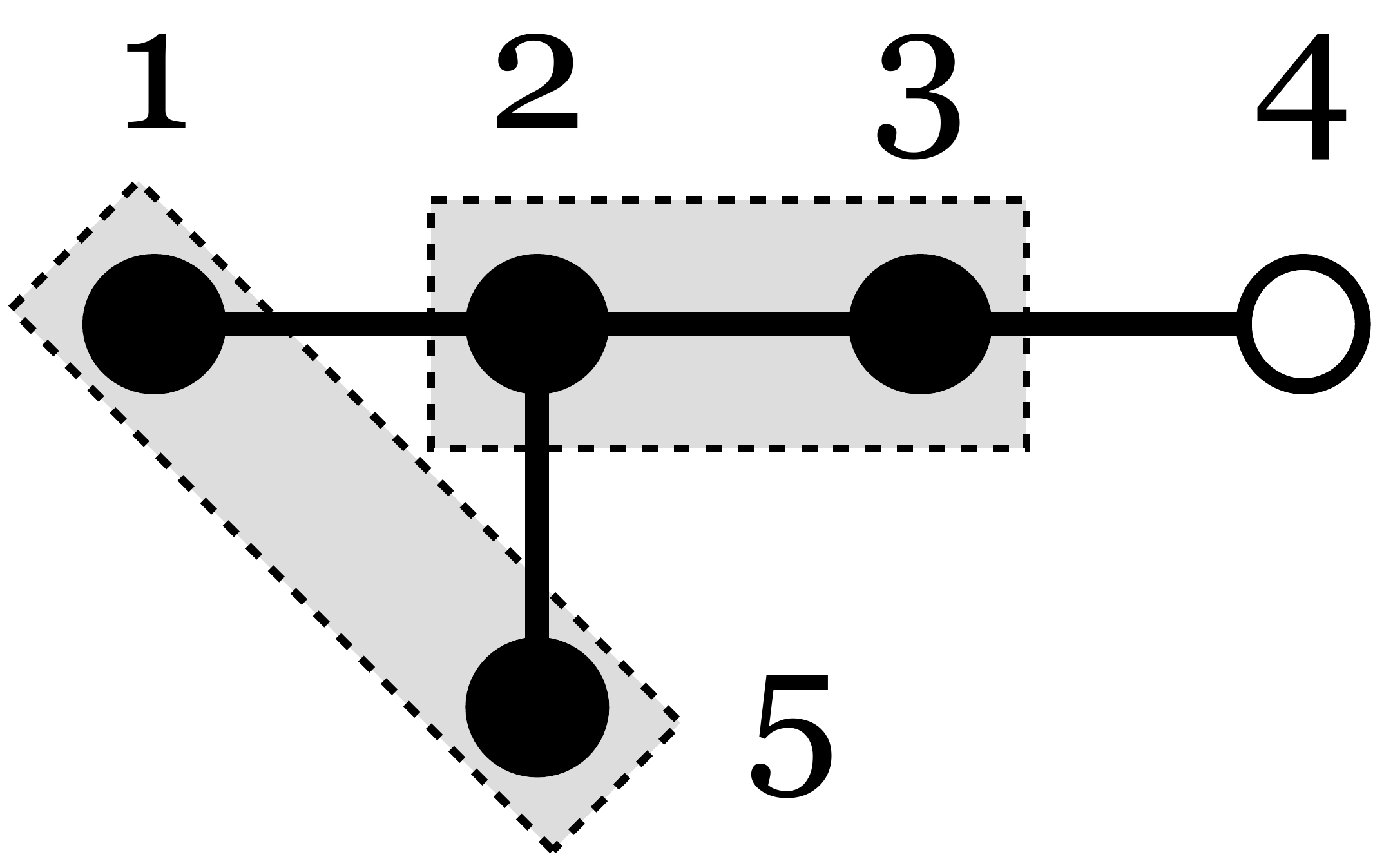}} \qquad
\subfloat[]{\includegraphics[width=0.13\textwidth]{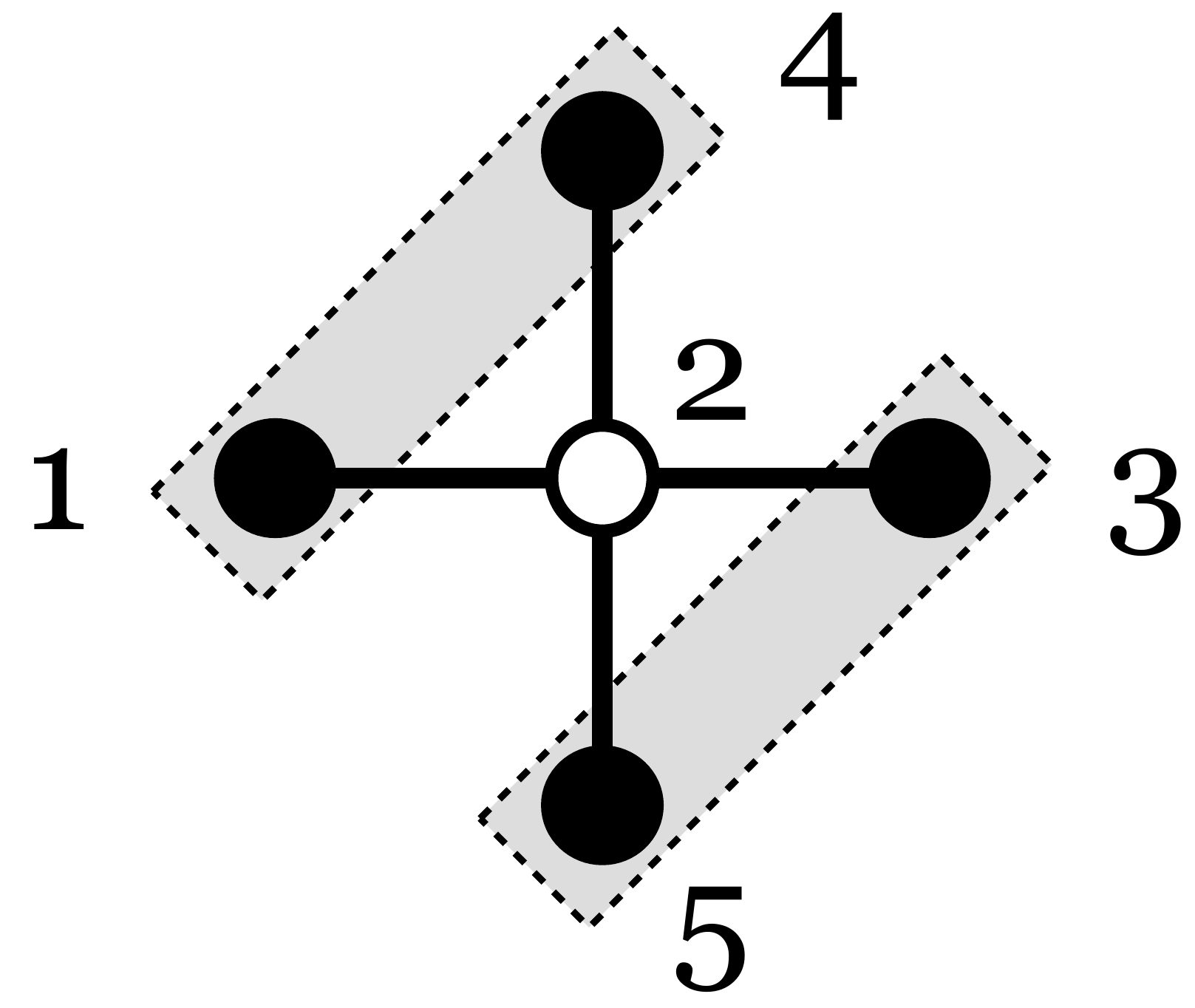}} \\
\caption{5-vertex graphs for implementing a universal set of logical two-qubit gates with the XY interaction. In all figures, gray boxes identify pairs of physical qubits that make up a logical qubit and white vertices represent ancillas initialized as $\ket{0}$.}
\label{fig:5vertex}
\end{figure} 

We first show how to implement a particular set of one- and two-qubit gates on the two 5-vertex graphs of \fig{5vertex}, similar to the simulation in \sec{examples} (cf.\ \fig{branching}). Suppose the two logical qubits can be initialized as in \fig{5vertex}a, according to the encoding of \eq{oddencoding}, together with one $\ket{0}$ ancilla. 

Since
\begin{equation}
H = \begin{pmatrix}
0 & 0 & 0 & 0 \\
0 & 0 & 2 & 0 \\
0 & 2 & 0 & 0 \\
0 & 0 & 0 & 0
\end{pmatrix},
\end{equation}
a logical $X$ rotation on the logical qubit stored in physical qubits $\{1,2\}$ can be implemented by a simple XY interaction:
\begin{equation}
\exp( i a X_L ) = \exp( i \tfrac{a}{2} H_{12} ) = \begin{pmatrix}
1 & 0 & 0 & 0 \\
0 & \cos{a} & i \sin{a} & 0 \\
0 & i \sin{a} & \cos{a} & 0 \\
0 & 0 & 0 & 1
\end{pmatrix}.
\end{equation}

We can also implement the two-qubit gate $R_{XZ}(a) := \exp(i a \, X \otimes Z)$ on the logical qubits $\{1,2\}$ and $\{3,4\}$ by the following sequence:
\begin{equation} \label{eq:2qubitg}
\is_{25} \, \is_{23} \, \is_{34} \, \left[ \is_{25}^\dagger \, \exp( i \tfrac{a}{2} H_{12} ) \, \is_{25} \right ] \, \is_{34}^\dagger \, \is_{23}^\dagger \, \is_{25}^\dagger.
\end{equation}
This sequence works as follows. The first three $\iswap$ gates use \eq{0iswap} to swap the qubits and place them as in \fig{5vertex}b. Notice that the first logical qubit suffers a $P$ gate during this operation. The sequence inside the square brackets implements an effective unitary with Hamiltonian $Y \otimes Z$. This can be verified by explicit multiplication, but can also be understood as follows: the $\is_{25}$ and $\is_{25}^{\dagger}$ swap qubits $2$ and $5$, leaving the first logical qubit encoded in pair $\{1,2\}$, up to some phases that depend upon the states of both qubits. The $H_{12}$ Hamiltonian then acts as a logical $X$ rotation on the first qubit. Keeping track of the dependence of the relative phases on the states of both qubits, we see that the overall operation is $Y \otimes Z$. Finally, the last three $\iswap$ gates return the states of all qubits to their original positions, while inducing a $P^{\dagger}$ gate on the first logical qubit. Since $P^{\dagger} Y P=X$, the overall operation on the encoded states is $X \otimes Z$, as claimed.

We now make a brief digression to explain why the set of Hamiltonians 
\begin{equation}
  \A := \{X \otimes I, I \otimes X, X \otimes Z, Z \otimes X, X \otimes Y, Y \otimes X \}
\end{equation}
is universal for quantum computation in the usual circuit model. First notice that the Hamiltonians $X \otimes Y$ and $Y \otimes X$ are included; this is without loss of generality, as they can be obtained as simple sequences of the remaining interactions, e.g., $X \otimes Y = U(X \otimes Z)U^{\dagger}$ where $U=\exp[ i \tfrac{\pi}{4} (I \otimes X)]$. By conjugating every element in $\A$ by $P$, we obtain the set 
\begin{equation}
  \B := \{Y \otimes I, I \otimes Y, Y \otimes Z, Z \otimes Y, X \otimes Y, Y \otimes X \}.
\end{equation}
These are exactly the generators of the special orthogonal group $\SO(4)$. This can be seen by writing them down explicitly, but also understood by a counting argument, as $\B$ contains six linearly independent, purely imaginary $4 \times 4$ matrices.

Now we recall the well-known fact (see, e.g., \cite{Bernstein93} and \cite{Rudolph02}) that universal quantum computation is possible using only orthogonal, rather than general unitary, matrices, with the overhead of one extra ancilla qubit and a polynomial number of operations. Furthermore, any special orthogonal matrix on $n$ qubits [i.e., in $\SO(2^n)$] can be decomposed in terms of $\SO(4)$ gates acting non-trivially only on pairs of qubits, so the set $\B$ is universal for quantum computation. But this means that the set $\A$ is also universal, since we can assume that initialization and measurements are done in the computational basis, so the initial and final single-qubit $\{P,P^{\dagger}\}$ gates do not affect the outcomes. 

While the graph in \fig{5vertex}a may not appear as a subgraph of the given graph, the sequence \eq{2qubitg} can be easily adapted to the graph of \fig{5vertex}c. In that case, we can just use \eq{0iswap} to swap the ancilla with any of the other qubits and obtain a similar arrangement to that of \fig{5vertex}b. The corresponding sequence is
\begin{equation} \label{eq:2qubitgb}
\is_{24} \, \left[ \is_{25}^\dagger \, \exp( i \tfrac{a}{2} H_{12} ) \, \is_{25} \right ] \, \is_{24}^\dagger.
\end{equation}
In this case, every operation described before is obtained up to conjugation by $P$, and the set of available operations is $\B$, rather than $\A$. However, as described above, this still suffices for universal computation. 

It remains to show that, for any graph other than a path or cycle, we can assign sufficiently many vertices as computational qubits and swap them around to one of the arrangements of \fig{5vertex} with a polynomial number of operations. 

\begin{theorem} \label{thm:maintheo2}
Let $G$ be any $n$-vertex connected graph, other than a path or a cycle, where every vertex represents a qubit and we can implement the interaction $H=X \otimes X + Y \otimes Y$ between any nearest neighbors in $G$. Then it is possible to efficiently simulate any quantum circuit on $\Omega(\sqrt{n})$ qubits.
\end{theorem}

\begin{proof}
As in \thm{maintheo}, it suffices to prove the universality of $H$ on any $n$-vertex tree $T$ that is not a path.

By \lem{graph}, either (i) the longest path of $T$ or (ii) the set of all its leaves must have more than $\sqrt{n}$ vertices. Suppose first that (i) holds. Then the universal construction is directly analogous to case (i) of \thm{maintheo}. Simply assign pairs of adjacent vertices on the longest path as logical qubits, and every other as a $\ket{0}$ ancilla. Then, by using \eq{logicswap2}, we can swap any two logical qubits to the closest degree-3 vertex, where we use sequence \eq{2qubitg} to implement the $X \otimes Z$ Hamiltonian as per \fig{5vertex}a. As explained previously, this together with the logical $X$ Hamiltonian on any qubit (given by $H$ on adjacent qubits) enables universal computation with overhead of at most $O(n)$ $\iswap$ operations per orthogonal matrix in the original circuit of \cite{Rudolph02}. 

Otherwise, (ii) holds. Then, first suppose that $T$ is not a star. Any such $T$ contains the graph of \fig{5vertex}a as a subgraph, so we assign those $5$ vertices as $\ket{0}$ ancillas, together with all non-leaves, and pair the remaining leaves arbitrarily into computational qubits. We can now use \eq{0iswap} to bring the states of any two logical qubits to the structure of \fig{5vertex}a, but with one caveat: this process may induce an overall $P$ gate on some logical qubits, depending on whether an odd or even number of $\ket{0}$ ancillas is traversed. This separates the logical qubits into two disjoint sets, namely those that suffer an overall $P$ gate and those that do not (there is no need to single out the case where the qubits suffer an overall $P^{\dagger}$, as this can be prevented by using $\iswap$$^{\dagger}$, rather than $\iswap$, as the swapping operation). We then take the larger of these two sets, which has at least $\sqrt{n}/4$ logical qubits, and for simplicity we disregard the rest. On the remaining qubits, as argued previously, we can either implement the set of operations
$\A$ or its conjugated-by-$P$ version 
$\B$. Since either set is universal, this gives an universal construction with an overhead of $O(\sqrt{n})$ operations for each gate in the original circuit.

Finally, for the star graph, we replace sequence \eq{2qubitg}, corresponding to \fig{5vertex}a, by the equivalent sequence \eq{2qubitgb} corresponding to \fig{5vertex}c. This enables us to implement the set of Hamiltonians mentioned in the previous paragraph, and concludes the proof.
\end{proof}

\section{Final remarks}

We have completely characterized the computational power of nearest-neighbor matchgates when the qubits are arranged on an arbitrary graph. This continues a line of research started in \cite{Brod12}, where the authors showed that matchgates are universal on many different graphs. Here we proved that the only connected graphs for which matchgates are classically simulable are paths and cycles, whereas on any other connected graph they are universal for quantum computation. Furthermore, we have shown that the same dichotomy holds when we restrict matchgates to the proper subset described by the XY interaction. This further expands the exploration of quantum computation with a single physical interaction \cite{DiVincenzo00,Kempe01b,Kempe02}, and could have applications for a variety of physical systems where the XY interaction arises naturally, if the placement of the qubits is subject to geometrical constraints. 

This dichotomy excludes the possibility that these two sets of interactions (general matchgates and the XY interaction), acting on graphs, could exhibit intermediate computational power such as that displayed by circuits of commuting observables (IQP) \cite{Bremner10} or noninteracting bosons \cite{Aaronson11}. However, this does not rule out such a result for other subsets of matchgates. As one example, consider the set generated by the $X \otimes X$ Hamiltonian acting on some graph. All such operations commute, and this set corresponds to a proper subclass of IQP. Furthermore, it was recently shown \cite{Hoban13} that the set of two-qubit $X \otimes X$ and single-qubit $X$ Hamiltonians are hard to simulate classically, modulo plausible complexity-theoretic assumptions, in the same way as IQP. We leave as open questions whether an analogous result can be obtained by further restricting the operation to only the $X \otimes X$ Hamiltonian, or possibly some other proper subset of matchgates, and how the power of such a model depends on the underlying interaction graph.

In our investigation we have not considered the use of non-trivial measurements to implement other unitary operations---it has been shown, for example, that noninteracting fermions (i.e., matchgates on a path) become universal if nondestructive charge measurements are allowed \cite{Beenakker04}. These charge measurements clearly cannot be implemented by combining matchgates and computational basis measurements. It might be interesting to consider other measurements and/or input states, beyond those obtainable by matchgates, and understand whether they change the computational capabilities of restricted subsets of matchgates on graphs.

While we have established universality of matchgates on any connected graph that is not a path or a cycle, it should be possible to improve the efficiency of our constructions. We have taken an operational approach, where each $\ket{0}$ is seen as an ``empty space'' through which we can move logical qubits, allowing for a simple and unified proof of universality for all graphs. In some cases, such as for the star graph, where all vertices but one are leaves, this construction is optimal. But in many others, our construction could ignore many vertices and/or edges, making it far from optimal. One such case is the binary tree of \fig{binarytree}, where we could have filled most of the non-leaves with logical qubits and used \eq{logicswap} rather than \eq{0swap} whenever it was necessary to ``move'' two logical qubits through each other. Since the bounds of \lem{graph} are tight (e.g., consider the graph obtained from a $\sqrt{n}$-leaf star by subdividing each edge $\sqrt{n}$ times), an optimal simulation presumably requires a more efficient assignment of logical qubits than in \thm{maintheo}. We believe that, while being markedly non-optimal in some cases, our construction nevertheless provides powerful tools for case-by-case optimization. We leave it as an open question whether there is a way to systematically obtain a more efficient construction, and in particular, whether in every case only a constant fraction of the qubits must be discarded as non-computational.

\acknowledgments
We thank Robin Kothari and Laura Man\v{c}inska for a helpful discussion of the proof of \lem{graph}, and Ernesto Galv\~ao for helpful discussions. D.B. would like to acknowledge financial support by Brazilian funding agency CNPq (Conselho Nacional de Desenvolvimento Cient\'ifico e Tecnol\'ogico). This work was also supported in part by NSERC, the Ontario Ministry of Research and Innovation, and the US ARO/DTO.

\end{document}